\documentclass[aps,pra,10pt,notitlepage,twocolumn]{revtex4-2}
\usepackage{amsfonts,amssymb,amsmath}            
\usepackage{lmodern}
\usepackage[]{graphics,graphicx}            
\usepackage{amsthm}
\usepackage{bbold,enumitem,mathtools}
\usepackage{tikz}
\usetikzlibrary{decorations.pathmorphing,patterns,decorations.markings,matrix,quantikz}
\usepackage{adjustbox}
\usepackage{nicematrix}
\usepackage{xstring}
\usepackage{ifthen}
\usepackage{float}
\usepackage{listings}
    \usepackage{hyperref}
\usepackage{cleveref}

\usepackage[english]{babel}

\makeatletter
\def\bbl@set@language#1{%
  \edef\languagename{%
    \ifnum\escapechar=\expandafter`\string#1\@empty
    \else\string#1\@empty\fi}%
  \@ifundefined{babel@language@alias@\languagename}{}{%
    \edef\languagename{\@nameuse{babel@language@alias@\languagename}}%
  }%
  \select@language{\languagename}%
  \expandafter\ifx\csname date\languagename\endcsname\relax\else
    \if@filesw
      \protected@write\@auxout{}{\string\select@language{\languagename}}%
      \bbl@for\bbl@tempa\BabelContentsFiles{%
        \addtocontents{\bbl@tempa}{\xstring\select@language{\languagename}}}%
      \bbl@usehooks{write}{}%
    \fi
  \fi}
\newcommand{\DeclareLanguageAlias}[2]{%
  \global\@namedef{babel@language@alias@#1}{#2}%
}
\makeatother

\DeclareLanguageAlias{en}{english}

\newtheorem{theorem}{Theorem}

\newtheorem{definition}{Definition}
\newtheorem{lemma}{Lemma}

\newenvironment{proof*}[1][\proofname]{%
  
  \begin{proof}[#1]}{\end{proof}}
\newcommand{\half}{\mbox{$\textstyle \frac{1}{2}$}}

\newcommand{\identity}{\mathbb{1}}
\renewcommand{\epsilon}{\varepsilon}

\pgfset{
  foreach/parallel foreach/.style args={#1in#2via#3}{evaluate=#3 as #1 using {{#2}[#3-1]}},
}

\pgfmathsetmacro\MathAxis{height("$\vcenter{}$")/2}

\begin{document}

\title{Encoded State Transfer: Beyond the Uniform Chain}
\date{\today}
\author{Alastair \surname{Kay}}
\affiliation{Royal Holloway University of London, Egham, Surrey, TW20 0EX, UK}
\email{alastair.kay@rhul.ac.uk}
\begin{abstract}
In a recent work \cite{kay2022}, we showed that a uniformly coupled chain could be symmetrically extended by engineered spin chains in such a way that we could choose part of the spectrum of the overall system. When combined with an encoding that avoids the uncontrolled eigenvalues, this resulted in the possibility of achieving a range of tasks such as perfect quantum state transfer. In this paper, we apply the same strategy to a much broader range of initial systems --- arbitrary chains, and even coupled networks of spins, while providing guarantees on the existence of solutions.
\end{abstract}
\maketitle

\section{Introduction}

The study of quantum state transfer, i.e.\ how to use the inherently local interactions of a device to communicate or create entanglement between distant qubits, is an important part of the puzzle of how to design modestly sized quantum computers. However, until our recent revelations in \cite{kay2022}, almost all studies were confined to a niche in which the design of the initial system had to be defined in advance, and stringent manufacturing requirement imposed \cite{christandl2004,karbach2005,kay2010a,vinet2012a}. Alternatives, such as \cite{burgarth2005a,wojcik2005,wojcik2007} circumvented some of these at the cost of imperfect transfer or indefinite arrival time. Now, these studies of perfect transfer are coming of age \cite{kay2022}. We see the potential to achieve perfect transfer by adapting an already manufactured system that we are given. It does not matter how the system is coupled; we can adapt. Nevertheless, \cite{kay2022} left open some crucial questions: does the construction apply beyond the uniformly coupled chain? Mathematically, solution of a linear problem was involved. How can we guarantee that such a problem has a solution? The purpose of this paper is to resolve these questions.

\begin{figure}
\centering
\includegraphics[width=0.45\textwidth]{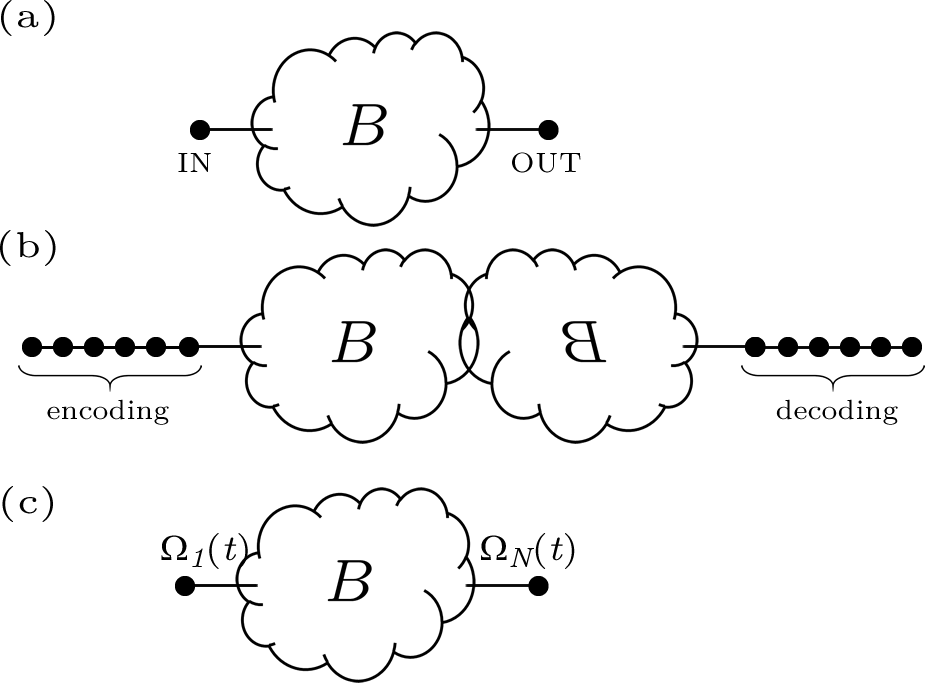}
\caption{(a) Given a fixed Hamiltonian, $B$, we aim to transfer states between two specific vertices. (b) $B$ is symmetrised about the output vertex, then the input vertex and its mirror are extended with chains. An encoding/decoding is used across the added regions. (c) The perfect transfer protocol is replicated with time control of the extremal couplings \cite{haselgrove2005}.}\label{fig:schematic}
\end{figure}

We will show how any given fixed system can be extended by engineered chains of spins (Fig.\ \ref{fig:schematic}) in order to facilitate tasks such as perfect state transfer and high quality state creation, a significant paradigm shift. To facilitate this, in Sec.\ \ref{sec:basics} we review some basic properties of spin systems and how the can be reconstructed from spectral information, followed by those of state transfer in Sec.\ \ref{sec:pst}. In Sec.\ \ref{sec:technical}, we perform most of the technical analysis, showing how any spin system can be extended so that part of the spectrum of the overall system is fixed to our demands, and how the other eigenvalues can be avoided via encoding. These results have certain conditions on success. In the following section, we analyse how our choice of (partial) target spectrum can ensure that those conditions are fulfilled. Finally, in Sec. \ref{sec:examples}, we demonstrate the application of the technique to a range of examples. These also provide the opportunity to numerically demonstrate the trade-offs if one does not want analytically perfect transfer and is willing to tolerate some error. Although we won't explicitly consider it, a recent extension of encoding strategy \cite{keele2021a} optimises noise tolerance, and is applicable in the present context.

\section{Mathematical Structure}\label{sec:basics}

Imagine we are given some fixed square matrix $B$ of dimension $N_B$. It will come with two clearly identified indices corresponding to `IN' and `OUT'. Without loss of generality, we number these as 1 and $N_B$ respectively. The matrix $B$ acts on a Hilbert space $\mathcal{B}$.

On any space $\mathcal{H}$, there will be a symmetry operator $S_{\mathcal{H}}$. In general, we need not precisely specify the symmetry operator, merely that $S_{\mathcal{H}}^2=\identity$ and that $S_{\mathcal{H}}$ acting on any index returns a different one. We say that $B$ is symmetric if $B=S_{\mathcal{B}}BS_{\mathcal{B}}$. In this case, $B$ decomposes into symmetric and antisymmetric components $B=B_+\oplus B_-$.

When we consider a \emph{tridiagonal} matrix $A$ in space $\mathcal{A}$, we shall define
$$
\mathcal{S}_A=\sum_{n=1}^{N_A}\ket{n}\bra{N+1-n}.
$$
If $A$ is symmetric, each of the diagonals is mirror-symmetric (also known as centrosymmetric).

Our aim will be to construct a matrix $H$ with a specific structure:
$$
H=\left(\begin{array}{ccc|ccc|ccc|ccc}
& & & & & & & & & & &\\
&S_\mathcal{A}AS_\mathcal{A}& & & & & & & & & &\\
& & &J& & & & & & & &\\\hline
& &J& & & & & & & & &\\\
& & & &B & & & & & & & \\
& & & & & &J' & & & & & \\\hline
& & & & & J' & & & & & & \\
& & & & & & & S_\mathcal{B}BS_\mathcal{B} & & & & \\
& & & & & & & & & J & & \\\hline
& & & & & & & & J & & & \\
& & & & & & & & & & A & \\
& & & & & & & & & & &
\end{array}\right).
$$
where $A$ is a tridiagonal matrix that we will choose in order to tune the spectrum of $H$ in a particular way. By design, $H$ is symmetric, yielding
$$
H_{\pm}=\left(\begin{array}{ccc|ccc}
& & & & & \\
&S_\mathcal{A}AS_\mathcal{A}& & & & \\
& & &J& & \\\hline
& &J& & & \\\
& & & &B\pm J'\proj{N_B} & \\
& & & & &
\end{array}\right).
$$
We shall typically denote $C_{\pm}=B\pm J'\proj{N_B}$.

Much of our study will be heavily dependent upon the characteristic polynomial $Q_H(z)$ of a given matrix $H$, which is monic. $P_H(z)$ is the characteristic polynomial of the principal submatrix of $H$, i.e.\ $H$ with its first row and column removed\footnote{The symmetry operator in the structure of $H_+$ allows us to refer to $P_A(z)$ and $P_B(z)$.}.
The spectrum of $H$ is denoted by $\text{spec}(H)$, obviously coinciding with the roots of $Q_H(z)$.

Crucial in what is to come is the relationship
\begin{equation}\label{eqn:overall characteristic}
\text{det}(z\identity-H_{\pm})=Q_{H_\pm}(z)=Q_A(z)Q_{C_\pm}(z)-J^2P_A(z)P_{C_{\pm}}(z).
\end{equation}

By Cauchy's interlacing theorem, the $N_B$ roots of $Q_{C_\pm}$ interlace with the $N_B-1$ roots of $P_{C_\pm}$. For tridiagonal matrices, the roots of $Q_{C_\pm}$ and $P_{C_\pm}$ \emph{strictly} interlace (i.e.\ there are no equalities) \cite{gladwell2005}.

\begin{lemma}\label{lem:evector inverse}
Consider a matrix of the form $H_{\pm}$, which has an eigenvalue $\lambda$. The eigenvector has the form $\left(\begin{array}{c} \ket{a} \\ \ket{b} \end{array}\right)$. Either $\lambda$ is an eigenvalue of $C_{\pm}$ and $\braket{N_A}{a}\braket{1}{b}=0$ such that $\ket{b}$ is the eigenvector of $C_{\pm}$, or
$$
\ket{b}\propto(C_{\pm}-\lambda\identity)^{-1}\ket{1}.
$$
\end{lemma}
\begin{proof}
We expand the eigenvector conditions to find
$$
C_{\pm}\ket{b}+J\proj{1}S_{\mathcal{A}}\ket{a}=\lambda\ket{b}.
$$
This rearranges to $(\lambda\identity-C_{\pm})\ket{b}=J\ket{1}\braket{N_A}{a}$. In the case that $\lambda$ is not in the spectrum of $C_{\pm}$, $\lambda\identity-C_{\pm}$ is invertible, and we can incorporate constants into normalisation to give $\ket{b}\propto(C_{\pm}-\lambda\identity)^{-1}\ket{1}$.

Otherwise, $C_{\pm}$ has an eigenvector $\ket{\lambda}$. Taking the inner product with this gives
$$
\bra{\lambda}(\lambda\identity-C_{\pm})\ket{b}=0=J\braket{\lambda}1\braket{N_A}{a}.
$$
If $\braket{\lambda}1=0$, then it is clear that $\ket{b}=\ket{\lambda}$ and $\ket{a}=0$ is a solution. Otherwise, if $\braket{N_A}{a}=0$, then $C_{\pm}\ket{b}=\lambda\ket{b}$, and hence $\ket{b}=\ket{\lambda}$. 

\end{proof}
The case of $\braket{\lambda}{1}=0$ is a special case that we refer to as an `unsupported state' \cite{pemberton-ross2010}. We will generally want to neglect these cases, and hence we refer to a ``fully supported $B$'' meaning the original $B$, but reduced in size by removing any such eigenvectors.

\begin{lemma}\label{lem:StrictInterlace}
For a fully supported matrix $C_{\pm}$, strict interlacing of the eigenvalues with the eigenvalues of the principal submatrix $C'$ always occurs.
\end{lemma}
\begin{proof}
Consider a fully supported matrix $C_{\pm}$ with the structure
$$
C_{\pm}=\left(\begin{array}{cc}
B_0 & \bra{J} \\
\ket{J} & C'
\end{array}\right),
$$
Let's assume that $C_{\pm}$ has an eigenvector of the form $\left(\begin{array}{c} \alpha \\ \ket{c} \end{array}\right)$ of eigenvalue $\lambda$ while $C'$ also has an eigenvalue $\lambda$ with eigenvector $\ket{\lambda}$. From the eigenvector condition on $C_{\pm}$, we have that
$$
-\alpha\ket{J}=(C'-\lambda\identity)\ket{c}.
$$
Taking the inner product with $\bra{\lambda}$ shows that $\braket{\lambda}{J}=0$ since $\alpha\neq 0$. However, in this case, there is a solution $\left(\begin{array}{c} 0 \\ \ket{\lambda} \end{array}\right)$, which we have explicitly excluded. As such, $C_{\pm}$ and $C'$ never have any eigenvalues in common. Thus, interlacing of eigenvalues becomes strict interlacing.
\end{proof}

\subsection{Sign Changes}

In \cite{gesztesy1997}, the following quantity was introduced:
$$\mu_A(z)=\bra{1}(z\identity-A)^{-1}\ket{1}=\frac{P_A(z)}{Q_A(z)}.$$
The final equality is readily verified by explicitly calculating the inverse -- the cofactor of the $(1,1)$ element of $z\identity-H$ is just $\text{det}(z\identity-H')/\text{det}(z\identity-H)$.

Let the set $\Sigma_{A}$ be the roots of $Q_{A}(z)$ and $P_{A}(z)$ in decreasing order. We note that as one varies $z$, $\mu_A(z)$ changes sign at the positions $\Sigma_A$. For a fully supported $A$, the positions of these sign changes are always distinct and hence the sign always alternates. Moreover, at large $z$, since $P_A$ and $Q_A$ are both monic, $\mu_A(z)$ must be positive. We find this useful to visualise, as in Fig.\ \ref{fig:tape}. We will later use these sign changes to help us confine roots of a rational equation to being within certain bounds.
\begin{figure}
\centering
\begin{tikzpicture}
\node [anchor=center] at (-0.7,0.5) {$\mu_A(z):$};
\draw [fill=blue!40,thick] (0,0) -- (7,0) to [bend left] (7,0.5) to [bend right] (7,1) -- (0,1) to [bend left] (0,0.5) to [bend right] (0,0);
\draw [draw=none,fill=red!40] (1,0) rectangle (1.5,1);
\draw [draw=none,fill=red!40] (3,0) rectangle (3.8,1);
\draw [draw=none,fill=red!40] (5.5,0) rectangle (6.8,1);
\draw [dashed,thick] (1,1.2) node [anchor=south] {$z_1$} -- (1,-0.2);
\draw [dashed,thick] (1.5,1.2) node [anchor=south] {$z_2$} -- (1.5,-0.2);
\draw [dashed,thick] (3,1.2) node [anchor=south] {$z_3$} -- (3,-0.2);
\draw [dashed,thick] (3.8,1.2) node [anchor=south] {$z_4$} -- (3.8,-0.2);
\draw [dashed,thick] (5.5,1.2) node [anchor=south] {$z_5$} -- (5.5,-0.2);
\draw [dashed,thick] (6.8,1.2) node [anchor=south] {$z_6$} -- (6.8,-0.2);
\node at (0.5,0.5) {+};
\node at (1.25,0.5) {-};
\node at (2.25,0.5) {+};
\node at (3.4,0.5) {-};
\node at (4.65,0.5) {+};
\node at (6.15,0.5) {-};
\draw [fill=none,thick] (0,0) -- (7,0) to [bend left] (7,0.5) to [bend right] (7,1) -- (0,1) to [bend left] (0,0.5) to [bend right] (0,0);
\draw [-latex,thick] (3,-0.5) -- (2,-0.5) node [anchor=east] {$z$};
\end{tikzpicture}
\caption{The value of the function $\mu_A(z)$ alternates sign at the distinct positions $z_i$ which satisfy $Q(z_{2i-1})=0$ and $P(z_{2i})=0$. The value at large $z$ is always positive.}\label{fig:tape}
\end{figure}
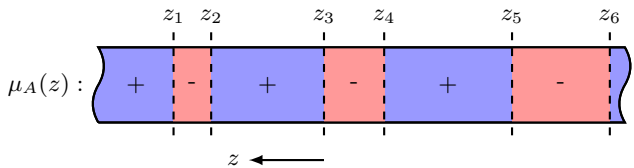

\subsection{Tridiagonal Matrices}\label{sec:key observations}

If an $N_A\times N_A$ Hermitian matrix $A$ is tridiagonal, then it has further useful properties. In particular, given $P_A$ and $Q_A$, we will be able to reconstruct it.
\begin{lemma}\label{thm:reconstruct}
$A$ is uniquely determined (under the assumption that all coupling strengths are positive) by $Q_A(z)$ and $P_A(z)$ if and only if $Q_A(z)$ and $P_A(z)$:
\begin{itemize}
\item have all real roots
\item have degrees $N_A$ and $N_A-1$ respectively
\item have strictly interlacing roots.
\end{itemize}
\end{lemma}
\begin{proof}
See \cite{gladwell2005}.
\end{proof}
Our challenge will be to specify the functions $P_A$ and $Q_A$. Ultimately, we will achieve this by fixing some of the eigenvalues of $H$. For now, we make two crucial observations to which we will return later. Firstly, if $H_+$ is to have an eigenvalue $\lambda$, then by Eq.\ (\ref{eqn:overall characteristic})
$$
J^2\mu_A(\lambda)\mu_{C_+}(\lambda)=1.
$$
Given $C_+$, we know $J^2\mu_A(\lambda)$. Secondly, consider a plot of $\mu_{C_+}(z)$ such as in Fig.\ \ref{fig:tape}. If we select two target eigenvalues $\lambda_1$ and $\lambda_2$ such that $\mu_{C_+}(\lambda_1)\mu_{C_+}(\lambda_2)<0$, then $\mu_A(z)$ must have a sign change in the region $(\lambda_1,\lambda_2)$, resulting from a root of either $P_A(z)$ or $Q_A(z)$, i.e.\ at least one of these has a root in the range $(\lambda_1,\lambda_2)$.

\section{Quantum State Transfer}\label{sec:pst}

We will correspond $H$ in the above definition with the single excitation subspace of a Hamiltonian of the form
\begin{equation}\label{eqn:exchange}
\frac12\sum_{i,j}J_{ij}(X_iX_j+Y_iY_j)+\frac12\sum_iB_iZ_i,
\end{equation}
i.e.\ the basis states are
$$
\ket{n}=\ket{0}^{\otimes (n-1)}\ket{1}\ket{0}^{\otimes(N-n)}.
$$
The matrix $H$ is the same as the adjacency matrix of the underlying coupling scheme for a graph $G$, where $J_{ij}$ define undirected edge weights between pairs of vertices/qubits, and $B_i$ are self-weights. We do not place any explicit restrictions on which qubits are coupled, merely assume that the $G$ is connected.

The task of state transfer is to identify a Hamiltonian of the form in Eq.\ (\ref{eqn:exchange}), such that if we initialise the entire system in $\ket{0}^{\otimes N}$, place an unknown state $\ket{\psi}$ on an input qubit, and wait time $t_0$, then the state arrives perfectly an a different, output site:
$$
\ket{\psi}_\text{IN}\ket{0}^{\otimes (N-1)}\longrightarrow\ket{0}^{\otimes (N-1)}\ket{\psi}_\text{OUT},
$$
up to a possible phase rotation. This task will be important in quantum computers where algorithms require two-qubit gates to be applied between arbitrary pairs of qubits, but the actual couplings available are only between nearest neighbours on some underlying graph. It is also the perfect test candidate for a range of new techniques that could have much broader applicability.

The study of perfect state transfer started in \cite{bose2003}, where it was observed that states could perfectly transfer through a chain of length 2 or 3. Perfect transfer is impossible in longer chains of uniform coupling \cite{bose2003,christandl2005}, but can be found if engineering of the coupling strengths is allowed \cite{christandl2004,kay2010a,karbach2005}. However, until \cite{kay2022}, this was a fundamental limitation --- it seemed necessary to impose extremely stringent requirements on the manufacture of a quantum device. \cite{kay2022} started to hint that this would not be necessary, by taking a chain that did not have perfect transfer, and imbuing it with that property. The goal of this paper is to show to breadth of initial systems to which this construction can be applied.

Imagine we start with a state $\ket{\Psi}$ in the single excitation subspace. We want to evolve it under a symmetric Hamiltonian $H$ for a time $t$. This Hamiltonian has eigenvalues $\lambda_n$ and eigenvectors $\ket{\lambda_n}$. Let us define
$
a_n=\braket{\lambda_n}{\Psi}.
$
After time $t$, we have a new state $\ket{\Psi(t)}=\sum_n\ket{\lambda_n}a_ne^{-i\lambda_nt}$. Our target state is $S_{\mathcal{H}}\ket{\Psi}$. Since $H$ is symmetric, $S_{\mathcal{H}}\ket{\lambda_n}=\pm\ket{\lambda_n}$. Thus, the conditions on perfect state transfer are that
$$
\exists t_0,\phi:\left\{\begin{array}{cc}
 e^{-i\lambda_n t_0}=e^{i\phi},& \forall n: a_n\neq 0,\lambda_n\in\text{spec}(H_+) \\
e^{-i\lambda_n t_0}=-e^{i\phi},& \forall n: a_n\neq 0,\lambda_n\in\text{spec}(H_-).
  \end{array}\right.
$$
Up to an arbitrary scale factor and offset, it is convenient to think of the spectrum for a perfect state transfer system as a set of integers where the even (odd) integers are assigned to $H_+$ ($H_-$) and the perfect transfer time is $\pi$. 

\begin{lemma}\label{lem:close}
Any symmetric tridiagonal matrix $A$ is arbitrarily close to a perfect state transfer solution.
\end{lemma}
\begin{proof}
This strategy is due to \cite{karbach2005}. Let $A$ have eigenvalues $\lambda_n$. Define an accuracy $\epsilon$, where $2\epsilon<\min_{i\neq j}|\lambda_i-\lambda_j|$. Within the range $\lambda_n/\epsilon\pm 1$, there is an even integer and an odd integer. Select the value $\eta_n$ to be the one with parity $(-1)^{n+1}$ (which ensures an assignment that is consistent with the parity of the eigenvector $\ket{\lambda_n}$). Reconstruct a symmetric tridiagonal matrix with eigenvalues $\{\epsilon\eta_n\}$ \cite{gladwell2005}. This new system has perfect transfer in a time $\pi/\epsilon$, and all eigenvalues are within $O(\epsilon)$ of the original matrix. By continuity (see Theorem 4 of \cite{hald1976}), the output matrix must be close to the original.
\end{proof}

\subsection{Encoding for Perfect Transfer}

\begin{figure}
\centering
\includegraphics[width=0.48\textwidth]{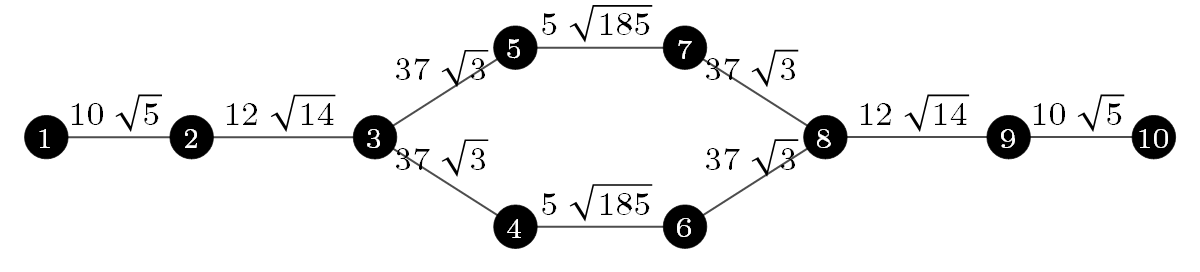}
\caption{Example spin network which is symmetric and has some eigenvalues that satisfy the perfect transfer conditions. Encoding is used to exclude the other eigenvectors. Vertices (numbered circles) correspond to qubits, while edges indicate a coupling with the given strength. $B_i=0$ for all vertices.}\label{fig:example}
\end{figure}

The conditions for perfect end-to-end transfer are well understood. Moving away from the ends of the chain is more complex, because we no longer know that all values $a_n$ must be non-zero, which potentially permits some eigenvalues to not match the transfer conditions and, in turn, this breaks the requirement that the chain be symmetric \cite{kay2010a,kay2011a}. Little progress \cite{kay2011a,coutinho2019} has been made in characterising these cases because it is difficult to coordinate non-satisfying eigenvalues to have eigenvectors that have no support on a specific site. Instead, we suggested \cite{kay2022} a mechanism that gives greater control: encoding over multiple input sites \cite{haselgrove2005} such that $a_n=0$ for all eigenvectors that cannot otherwise satisfy the eigenvalue conditions. If there are $M-1$ of these, this is achieved with an encoding region of size $M$ by finding the null space of the $\{\ket{\lambda_n}\}$ restricted to the encoding region. 

Consider, for example, the network depicted in Fig.\ \ref{fig:example}. This has eigenvalues $\pm\sqrt{185}\times\{1,2,5,6,10\}$. Without the $\pm 1,5$ terms, there would be perfect state transfer in the time $t=\pi/(4\sqrt{185})$. Next, we note that the eigenvectors with eigenvalue $\pm5\sqrt{185}$ only have support on vertices 4,5,6,7. On the encoding region (first 3 sites) the two eigenvectors corresponding to $\pm\sqrt{185}$ have components
$$
\ket{1}\mp\frac{\sqrt{37}}{10}\ket{2}-\frac{3}{8}\sqrt{\frac{7}{10}}\ket{3}.
$$
Orthogonal to these (and the $\pm5\sqrt{185}$ eigenvectors) is
$$
3 \sqrt{7}\ket{1}+ 8\sqrt{10}\ket{3},
$$
which in time $t=\pi/(4\sqrt{185})$ therefore transfers onto the decoding region (also size 3) as
$$
-i(3 \sqrt{7}\ket{10}+ 8\sqrt{10}\ket{8}).
$$
We have perfect encoded transfer of a single excitation.

This interpretation of encoding/decoding regions provides the opportunity to incorporate encoding into our analytic strategies, rather than just adding it in subsequently. Moving beyond the perfect transfer regime, we could encode against the `worst offender' eigenvalues.

\section{Reconstructing a matrix from mixed structural and spectral information}\label{sec:technical}

The setting is that we are given a matrix $B$. We can use this to build a matrix of the form of $H$. We do not yet know the values of $J$ or $J'$, or any of the properties of $A$, except that we want it to be tridiagonal. We can use these various parameters to help us create perfect encoded state transfer. In fact, we will not use the $J'$ freedom at all, and just fix it to an arbitrary value. We will allocate the eigenvalues of $H$ (which we don't yet know) two two sets: $\Gamma_P$ and $\Gamma_{\bar P}$. Those that match the perfect state transfer condition at time $t_0$ go in $\Gamma_P$. Those that do not are allocated to $\Gamma_{\bar P}$.The previous section conveyed that provided the size of the encoding region is larger than $|\Gamma_{\bar P}|$, perfect encoded transfer is possible\footnote{The encoding region might be considered to contain $N_A$ or $N_A+1$ qubits. It's a matter of taste/definition, and we give examples that incorporate both options.}.


Our aim is to discover how to specify the system $A$, giving us control over some of the eigenvalues so that we can force them to satisfy the perfect transfer conditions. In fact, everything from \cite{kay2022} may be directly imported.

If we want $H$ to have a particular eigenvalue $\lambda$ in the subspace $\sigma$, then it must satisfy $Q_{H_\sigma}(\lambda)=0$. By Eq.\ (\ref{eqn:overall characteristic}), this simplifies to
$
J^2\mu_A(\lambda)\mu_{C_\sigma}(\lambda)=1.
$

\begin{theorem}\label{thm:fixedmiddle}
Given the known matrices $C_{\pm}$, along with two sets of eigenvalues $\Lambda_+=\{\lambda_n^+\}$ and $\Lambda_-=\{\lambda_n^-\}$, we can construct the unique solution for $H$ with eigenvalues $\Lambda_+$ in the symmetric subspace and $\Lambda_-$ in the antisymmetric subspace, satisfying $|\Lambda_+|+|\Lambda_-|=2N_A-1$ ($J$ known) or $|\Lambda_+|+|\Lambda_-|=2N_A$ ($J$ unknown), if the solution exists.
\end{theorem}
\begin{proof}
We need the function $\mu_A(z)$ to satisfy
\begin{align*}
J^2\mu_A(\lambda)\mu_{C_+}(\lambda)=1&\quad\forall \lambda\in\Lambda_+, \\
J^2\mu_A(\lambda)\mu_{C_-}(\lambda)=1&\quad\forall \lambda\in\Lambda_-.
\end{align*}
If $J$ is known, then $J^2\mu_{C_{\pm}}$ can be evaluated for each value of $\lambda_n^{\pm}$, determining $\mu_A(z)$ at various positions. Since
$$
\mu_A(z)=\frac{P_A(z)}{Q_A(z)}=\frac{\sum_{i=0}^{N_A-1}a_iz^i}{\sum_{i=0}^{N_A}b_iz^i}
$$
where $a_{N_A-1}=b_{N_A}=1$, this is a linear problem in the coefficients $\{a_i\}$ and $\{b_i\}$. Explicitly, one must solve
\begin{widetext}
\begin{equation}\label{eqn:jacobistyle}
\left(\begin{array}{cc}
\displaystyle\sum_{n=1}^{|\Lambda_+\cup\Lambda_-|}\sum_{i=0}^{N_A-2}\lambda_n^i\ket{n}\bra{i} & -\displaystyle\sum_{n=1}^{|\Lambda_+\cup\Lambda_-|}\sum_{i=0}^{N_A-1}\mu_A(\lambda_n)\lambda_n^i\ket{n}\bra{i}
\end{array}\right)\left(\begin{array}{c} \vec{a} \\ \vec{b} \end{array}\right)=
\sum_{n=1}^{|\Lambda_+\cup\Lambda_-|}(\mu_A(\lambda_n)\lambda_n^{N_A}-\lambda_n^{N_A-1})\ket{n}.
\end{equation}
\end{widetext}
If, instead, $J$ is unknown, then we define
$$
J^2\mu_A(z)=\frac{J^2P_A(z)}{Q_A(z)}=\frac{\sum_{i=0}^{N_A-1}a_iz^i}{\sum_{i=0}^{N_A}b_iz^i}
$$
where $a_{N_A-1}=J^2$ instead of 1. We solve an equivalent problem to before,
and instead return $J=\sqrt{a_{N_A-1}}$ and
$$
P_A(z)=\frac{1}{a_{N_A-1}}\sum_ia_iz^i.
$$

Assuming there is a solution, it is unique, and in principle it is readily found (there are numerical challenges since the matrix that needs to be inverted is closely related to a Vandermonde matrix, and suffers from the same numerical instabilities). Once we have the polynomials, Lemma \ref{thm:reconstruct} yields $A$ and hence $H$.
\end{proof}

The only element to be added to the presentation of \cite{kay2022}, which only explicitly considered a uniformly coupled chain for the system $B$, was to add was the symmetrisation of $B$. The numerical tests in \cite{kay2022} were extremely successful. Essentially, this is because the model is very forgiving. Now, however, we must deal with a wide range of situations in $B$, including extremes such as Anderson localised chains.  Do solutions always exist? How are we to select target eigenvalues to ensure a solution exists? What are the consequences for the transfer time?

We are now in a position to state our central thesis, which we will build towards proving in the next section.
\begin{theorem}\label{thm:centralthesis}
For every $B$ such that the eigenvalues of $C_+$ are distinct from those of $C_-$ (for any $J'$), there exist extensions $A$ of tridiagonal form that permit perfect (encoded) state transfer.
\end{theorem}

For a chain (i.e.\ tridiagonal $B$), the eigenvalues are always distinct, as justified by the strict interlacing property of the eigenvalues for the symmetric central block
$$
C=\left(\begin{array}{ccc|ccc}
&&&&& \\
& B &&&& \\
&&&J'&& \\\hline
&&J'&&& \\
&&&&S_{\mathcal{B}}BS_{\mathcal{B}} & \\
&&&&&
\end{array}\right).
$$

We further note that if the addition of so many vertices seems excessive, one could instead apply another result from \cite{haselgrove2005}, which allows one to replace the entire length of the added system $A$ with a single qubit and time control of the coupling $J$ (see Fig.\ \ref{fig:schematic}). In that context, Theorem \ref{thm:centralthesis} is equivalent to reachability results in control theory \cite{pemberton-ross2010}. Unlike control theory, this is a constructive result, revealing exactly how to do the transfer, and providing insights about how long that will take.

\section{Towards Existence}

The great challenge for us is to find cases that work. If a solution exists, it's unique, but when does a solution exist? There are three key conditions (Theorem \ref{thm:reconstruct})  which, taken together, are necessary and sufficient:
\begin{enumerate}
\item All the roots of $P_A$ and $Q_A$ must be real.
\item $J$ must be real ($J$ unknown) or $P_A$ must be monic ($J$ known).
\item The roots of $Q_A$ and $P_A$ must strictly interlace. 
  \end{enumerate}
In this section, we will show how to specify the (partial) target spectrum $\Lambda$ in order to guarantee that these conditions are fulfilled.

We have already observed in Sec.\ \ref{sec:key observations} that we can ensure all the roots of $P_A$ and $Q_A$ are real: there are $2N_A-1$ roots in total, and we are selecting $2N_A$ target eigenvalues. Provided we select them such that the sign of $\mu_A(z)$ alternates, then we have pinned the positions of all the roots in the real range. Moreover, we will fix $J$ to be real by ensuring we start with the correct sign of $\mu_A(z)>0$ for the largest target eigenvalue.

Note that there is no limitation here in terms of (i) choosing the eigenvalues to satisfy the perfect transfer conditions (follow a construction akin to \ref{lem:close}), or (ii) the number of target eigenvalues that we can pick. If we compare the signs of $\mu_{C_\pm}(z)$, then there will be ranges of $z$ where both functions have the same sign. We will be able to select at most one target value of $\lambda$ in such a range. However, there will also be ranges (every second one assuming $\Sigma_{C_+}\cap\Sigma_{C_-}=\emptyset$) for which their signs are different. In such a range, one can pick as many target eigenvalues as desired: you just have to alternate the target symmetry. This is depicted in Fig.\ \ref{fig:double tape}.

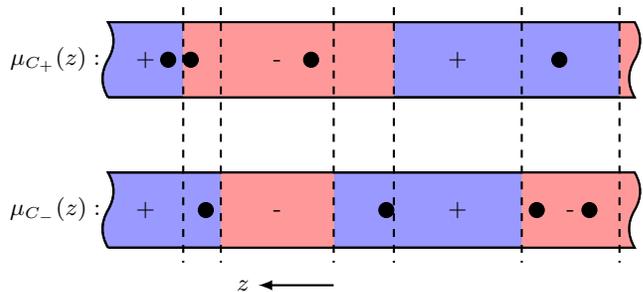
\begin{figure}
\centering
\begin{tikzpicture}
\begin{scope}
\node [anchor=center] at (-0.7,0.5) {$\mu_{C_+}(z):$};
\draw [fill=blue!40,thick] (0,0) -- (7,0) to [bend left] (7,0.5) to [bend right] (7,1) -- (0,1) to [bend left] (0,0.5) to [bend right] (0,0);
\draw [draw=none,fill=red!40] (1,0) rectangle (3,1);
\draw [draw=none,fill=red!40] (3,0) rectangle (3.8,1);
\draw [draw=none,fill=red!40] (6.8,0) -- (7,0) to [bend left] (7,0.5) to [bend right] (7,1) -- (6.8,1) -- cycle;
\node at (0.5,0.5) {+};
\node [yshift=-\MathAxis] at (2.25,0.5) {-};
\node at (4.65,0.5) {+};
\draw [fill=none,thick] (0,0) -- (7,0) to [bend left] (7,0.5) to [bend right] (7,1) -- (0,1) to [bend left] (0,0.5) to [bend right] (0,0);
\end{scope}
\begin{scope}[yshift=-2cm]
\node [anchor=center] at (-0.7,0.5) {$\mu_{C_-}(z):$};
\draw [fill=blue!40,thick] (0,0) -- (7,0) to [bend left] (7,0.5) to [bend right] (7,1) -- (0,1) to [bend left] (0,0.5) to [bend right] (0,0);
\draw [draw=none,fill=red!40] (1.5,0) rectangle (3,1);
\draw [draw=none,fill=red!40] (5.5,0) -- (7,0) to [bend left] (7,0.5) to [bend right] (7,1) -- (5.5,1) -- cycle;
\draw [dashed,thick] (1,3.2) -- (1,-0.2);
\draw [dashed,thick] (1.5,3.2) -- (1.5,-0.2);
\draw [dashed,thick] (3,3.2) -- (3,-0.2);
\draw [dashed,thick] (3.8,3.2) -- (3.8,-0.2);
\draw [dashed,thick] (5.5,3.2) -- (5.5,-0.2);
\draw [dashed,thick] (6.8,3.2) -- (6.8,-0.2);
\node at (0.5,0.5) {+};
\node [yshift=-\MathAxis] at (2.25,0.5) {-};
\node at (4.65,0.5) {+};
\node [yshift=-\MathAxis] at (6.15,0.5) {-};
\draw [fill=none,thick] (0,0) -- (7,0) to [bend left] (7,0.5) to [bend right] (7,1) -- (0,1) to [bend left] (0,0.5) to [bend right] (0,0);
\end{scope}
\draw [-latex,thick] (3,-2.5) -- (2,-2.5) node [anchor=east] {$z$};
\filldraw (0.8,0.5) circle (0.1cm);
\filldraw (1.1,0.5) circle (0.1cm);
\filldraw (1.3,-1.5) circle (0.1cm);
\filldraw (2.7,0.5) circle (0.1cm);
\filldraw (3.7,-1.5) circle (0.1cm);
\filldraw (5.7,-1.5) circle (0.1cm);
\filldraw (6,0.5) circle (0.1cm);
\filldraw (6.4,-1.5) circle (0.1cm);
\end{tikzpicture}
\caption{Sign of the values of the functions $\mu_{C_{\pm}}(z)$. Dashed lines indicate positions of a member of $\Sigma_{C_{\pm}}$, inducing a sign change in the $\mu$ of the same symmetry. In bands where both $\mu_{C_{\pm}}(z)$ have the same sign, select a maximum of one target eigenvalue. In regions where they differ, select as many target eigenvalues as desired, provided they alternate in symmetry. A possible selection is depicted by black circles, with symmetry indicated by vertical position. As one changes $z$, they alternate sign, starting with a positive sign at large $z$.}\label{fig:double tape}
\end{figure}

So, by judicious choice of the target spectrum, we can guarantee that all the roots of $Q_A$ and $P_A$ are real, and that $J$ is real, even if we don't (yet) know anything about the interlacing of the roots.

\subsection{Interlacing of \texorpdfstring{$P$}{P} and \texorpdfstring{$Q$}{Q}}

Of the three necessary and sufficient conditions,  we have shown that two can be simultaneously satisfied, and are compatible with the conditions for perfect state transfer. It remains to ensure that the roots of $P$ and $Q$ interlace. Predictable interlacing properties are known to be a particular challenge \cite{maione2011}. 

We introduce a strategy which we refer to as ``pair pinning'', in which we place target eigenvalues either side of a zero of $Q_C$. Provided they are close enough that there is no interference from other members of $\Sigma_{C_{\pm}}$, the two targets yield opposite signs for $\mu_A(z)$, helping to force a root or pole. One might have guessed that this would force a 0 of $P_A$ as our chosen points move closer to the zero. This is certainly not always the case, as conveyed by the case where we \emph{only} specify two target eigenvalues, meaning that $Q_A$ has degree 1 and $P_A$ has degree 0. If there is a solution, the root \emph{must} be a root of $Q_A$. So, select any $\sigma$ that is a root of $Q_C$ and place $\lambda_1,\lambda_2$ either side of it. Interlacing is trivially satisfied, and we have guaranteed the other two necessary conditions. 

Let us generalise this to a broader class of examples.
\begin{lemma}\label{example}
Consider two sets of points $\underline{\lambda}^{\pm}$, both of size $N$ such that $\lambda^+_i>\lambda^-_i>\lambda^+_{i+1}$. If $\mu(z)=P_A(z)/Q_A(z)$ satisfies $\mu(\lambda^\sigma)=\sigma\epsilon$ for $\lambda^{\sigma}\in\underline{\lambda}^{\sigma}$ and constant $\epsilon$, then $Q_A(z)$ has roots in the ranges $(\lambda_i^-,\lambda_i^+)$. 
\end{lemma}
\begin{proof}
Consider the linear problem of finding $P_A(z)=\sum_{n=0}^{N-1}a_nz^n$ and $Q_A(z)=\sum_{n=1}^{N-1}b_nz^n+z^{N}$ that satisfy
$$
P_A(\lambda^{\sigma})=Q_A(\lambda^{\sigma})\epsilon\sigma.
$$
In other words,
$$
\left(\begin{array}{cc}
V_+ & -V_+ \\ V_- & V_-
\end{array}\right)\left(\begin{array}{c} \underline{a} \\ \epsilon\underline{b}\end{array}\right)=\epsilon\left(\begin{array}{c}
{\underline{\lambda}^+}^N \\
-{\underline{\lambda}^-}^N
\end{array}\right)
$$
where the Vandermonde matrices $V_{\pm}$ are formed from the elements of $\underline{\lambda}^{\pm}$,
$$
V=\sum_{n=1}^N\sum_{m=0}^{N-1}\lambda_n^m\ket{n}\bra{m}.
$$
The values of $Q_A(z)$ at positions $\underline{\lambda}^{\pm}$ are evaluated by
\begin{multline*}
\left(\begin{array}{c}
Q_A(\underline{\lambda}^+) \\
Q_A(\underline{\lambda}^-)
\end{array}\right)=
\left(\begin{array}{cc}
0 & V_+ \\
0 & V_-
\end{array}\right)\left(\begin{array}{c} \underline{a} \\ \underline{b}\end{array}\right)+\left(\begin{array}{c}
{\underline{\lambda}^+}^N \\
{\underline{\lambda}^-}^N
\end{array}\right)\\=\frac{1}{2}\left(\begin{array}{cc}
\identity & -V_+V_-^{-1} \\
-V_-V_+^{-1} & \identity
\end{array}\right)\left(\begin{array}{c}
{\underline{\lambda}^+}^N \\
{\underline{\lambda}^-}^N
\end{array}\right).
\end{multline*}

Consider the term $\vec{c}=-V_-^{-1}{\underline\lambda^-}^N$. Since $V_-\vec{c}+{\underline\lambda^-}^N=0$, $\vec{c}$ are the coefficients (of degrees 0 to $N-1$) of the degree $N$ monic polynomial with roots at $\underline{\lambda}^-$. Hence, $V^+\vec{c}+{\underline\lambda^+}^N$ (the top row of the matrix), evaluates that polynomial at the points $\underline{\lambda}^+$. In other words,
$$
Q_A(\lambda_i^+)=\prod_n(\lambda_i^+-\lambda_n^-),\qquad Q_A(\lambda_i^-)=\prod_n(\lambda_i^--\lambda_n^+).
$$
Given how these values interlace, the only terms that are different in sign are $(\lambda_i^+-\lambda_i^-)$ and $(\lambda_i^--\lambda_i^+)$ respectively: the two values of $Q_A$ have opposite signs. $Q_A$ has a root in the range $(\lambda_i^-,\lambda_i^+)$ for each $i$.
\end{proof}

We are finally in a position to prove the central claim of Theorem \ref{thm:centralthesis} that (almost) any symmetric fully supported $C$ can be symmetrically extended to provide perfect (encoded) state transfer.
\begin{proof}[Proof of Theorem \ref{thm:centralthesis}]
Let $\{\lambda_n\}$ be all the roots of $Q_{C_+}$ and $Q_{C_-}$, in descending order, and let $\sigma_n$ be such that $Q_{C_{\sigma_n}}(\lambda_n)=0$. We select a $\delta$ that is small enough such that $\lambda_n-\delta>\lambda_{n+1}+\delta$, and such that the appropriate polynomial $P_{C_{\sigma_n}}(z)$ does not change sign in the range $\lambda_n\pm\delta$ (this region is non-trivial by Lemma \ref{lem:StrictInterlace}), provided $Q_{C_{\pm}}$ do not share any roots.

By this choice, $f(z)=\frac{Q_{C_{\sigma_n}}(z)}{P_{C_{\sigma_n}}(z)}$ is continuous on the ranges $\lambda_n\pm\delta$, being positive in the range $(\lambda_n,\lambda_n+\delta)$, and negative in the range $(\lambda_n-\delta,\lambda_n)$. It is therefore possible to select values $\eta^{\pm}_n$ from within each of those ranges such that $f(\eta^{\pm}_n)=\pm\epsilon$ for some sufficiently small $\epsilon$. We know that a solution for this exists by Lemma \ref{example}. Moreover, this solution must be arbitrarily close to a perfect transfer solution via Lemma \ref{lem:close}. Hence by continuity, there must exist a perfect transfer solution.
\end{proof}
This result merely proves the existence of a solution, paying little heed of the transfer time, which will have large overheads due to several limits being taken. In our numerical examples, we ignore the need to force all the $\mu_A(z)$ values to be the same, and jump straight to finding the best matching perfect transfer eigenvalues. We have found this strategy to be extremely effective.

\subsection{Transfer Speed}

Our selection of eigenvalues via this pair pinning strategy reveals something about the state transfer speed. We supply two target eigenvalues for every root of $Q_C$. Hence, if the smallest gap in the spectrum of $C$ is $\Delta$, the smallest gap in $H$ is no larger than $\Delta$, and the state transfer time must take at least $\pi/\Delta$. This bounding of the state transfer time is important to understand -- what if we had started with a chain that had such strong magnetic fields, for example, that it was exhibiting Anderson localisation \cite{anderson1958}, implying that there is no transfer? It \emph{does not} imply that the transfer is impossible, just that it takes a long time \cite{burrell2007}, and we cannot expect to overcome that long time restriction by extending the chain --- one characterisation of Anderson localisation is that, in the long chain limit, there is a band where the spectrum is continuous. For a finite chain, these vanishingly small energy gaps must therefore yield extremely long transfer times. Localisation is retained, but the system does achieve perfect transfer, just very slowly (presumably, exponential in the chain length \cite{burrell2007}).

\section{Examples}\label{sec:examples}

\subsection{The Field-Free Case}

For the purposes of numerical examples, it will often be convenient to work with the field-free case, i.e.\ where the matrices $A$ and $B$ have 0 on the diagonal, as this halves the number of parameters involved.

\begin{definition}
Let $H$ be a symmetric Hermitian matrix that describes couplings on an underlying graph that is bipartite (two-colourable). Further, require that the diagonal elements are 0. Finally, impose that $S_\mathcal{H}\ket{i}$ yields a vertex of the opposite `colour' as vertex $i$ in the graph. We say that $H$ is field-free and even. 
\end{definition}

The `even' feature is built into our assumed structure of $H$, so this just requires that $B$ is bipartite.

\begin{lemma}
If $H$ is field-free and even, then
$$
\Lambda^-=\{-\lambda:\lambda\in\Lambda^+\}.
$$
\end{lemma}
\begin{proof}
Since $H$ is bipartite, there is a partition of the vertices of $H$, $V$ such that $V=V_1\cup V_2$, $V_1\cap V_2=\emptyset$ and $H_{ij}=0$ for all choices of $i$ and $j$ that are in the same partition. As such, we can define
$$
D=\sum_{i\in V_1}\proj{i}-\sum_{i\in V_2}\proj{i}
$$
from which it follows that $D^2=\identity$ and $DHD=-H$. In other words, if $\lambda$ is an eigenvalue of $H$, then so is $-\lambda$.

Let's assume that $\ket{\lambda}$, the eigenvector associated with eigenvalue $\lambda$ is symmetric, i.e.\ $S_\mathcal{H}\ket{\lambda}=\ket{\lambda}$. Now, $\ket{-\lambda}=D\ket{\lambda}$. So, $S_\mathcal{H}\ket{-\lambda}=S_\mathcal{H}DS_\mathcal{H}^2\ket{\lambda}=S_\mathcal{H}DS_\mathcal{H}\ket{\lambda}$. $H$ is even by assumption, meaning that $S_\mathcal{H}DS_\mathcal{H}=-D$. Thus, $\ket{-\lambda}$ is anti-symmetric. 
\end{proof}

\begin{lemma}\label{lem:subsapces}
If the matrix $C$ describes couplings on an underlying graph that is bipartite (two-colourable) and field-free, then the polynomial $\mu_A(z)$ is antisymmetric.
\end{lemma}
\begin{proof}
If $C$ is bipartite, and $A$ is a field-free tridiagonal matrix, then $H$ overall is bipartite. Let $\tilde D$ be the restriction of $D$ to just the $C$ matrix such that $\tilde DC\tilde D=-C$.

Let's say we want to impose that $H$ has a specific eigenvalue $\lambda$. It must also have an eigenvalue $-\lambda$. Let us take the case that $\lambda$ is in the spectrum of $H_\sigma$ where $\sigma\in\pm$. The parity of the $-\lambda$ eigenvector is then fixed because $\ket{-\lambda}=D\ket{\lambda}$. Hence, we will be forced to select $-\lambda$ to be in the spectrum of $H_\tau$, $\tau\in\pm$. We thus require
$$
J^2\mu_A(\lambda)\mu_{C_\sigma}(\lambda)=J^2\mu_A(-\lambda)\mu_{C_\tau}(-\lambda)=1.
$$
Our statement is equivalent to $\mu_{C_\sigma}(\lambda)=-\mu_{C_\tau}(-\lambda)$.

We now introduce a vector
$$
\ket{s_{\pm}}=\frac{1}{\sqrt{2}}(\ket{1}\pm S_{\mathcal{C}}\ket{1}).
$$
Note that $\tilde D\ket{s_\sigma}=\ket{s_\tau}$. We can rewrite
\begin{align*}%
\mu_{C_{\sigma}}(\lambda)&=\bra{1}(\lambda\identity-C_{\sigma})^{-1}\ket{1} \\
&=\bra{s_\sigma}(\lambda\identity-C)^{-1}\ket{s_{\sigma}} \\
&=\bra{s_\sigma}(\lambda \tilde D^2+\tilde DC\tilde D)^{-1}\ket{s_{\sigma}} \\
&=-\bra{s_\sigma}\tilde D((-\lambda) \identity-C)^{-1}\tilde D\ket{s_{\sigma}} \\
&=-\bra{s_\tau}((-\lambda) \identity-C)^{-1}\ket{s_{\tau}} \\
&=-\mu_{C_\tau}(-\lambda).
\end{align*}
Hence, the conditions on $\mu_A$ are antisymmetric.
\end{proof}
\noindent  In the case of a symmetric and even matrix, $\tau=-\sigma$. One special case that this scenario covers is a chain -- we are given a symmetric field-free chain of even length and are trying to `complete' it by adding a field-free chain at either end. If we specify that all the eigenvalues we want should occur in $\pm\lambda$ pairs then since the overall size of the system is even, if $\lambda$ is an eigenvalue of $H_+$ then $-\lambda$ is an eigenvalue of $H_-$, and we know the corresponding values of $\mu$ should be opposite. This in turn means that the construction of $A$ will be field-free, and hence the overall matrix is field-free. Moreover, the fact that $\mu_A(z)$ is antisymmetric means that we should be able to specify it using only half the data points (e.g.\ just using those $\lambda$ that are in the spectrum of $H_+$). As indicated in \cite{kay2022}, this structure can be built in. In order to find an antisymmetric rational function $f(z)$ that satisfies $f(z_i)=f_i$ and $f(-z_i)=-f_i$, it is sufficient to find the rational function $g(z)$ which satisfies $\{g(z_i^2)=f_i/z_i\}$, such that if $g(z)=\dfrac{P(z)}{Q(z)}$, then \begin{equation}f(z)=\dfrac{zP(z^2)}{Q(z^2)}.\label{eq:sub}\end{equation}

\subsection{Central region with random errors} \label{sec:noisyuniform}

\begin{figure}
\centering
\includegraphics[width=0.45\textwidth]{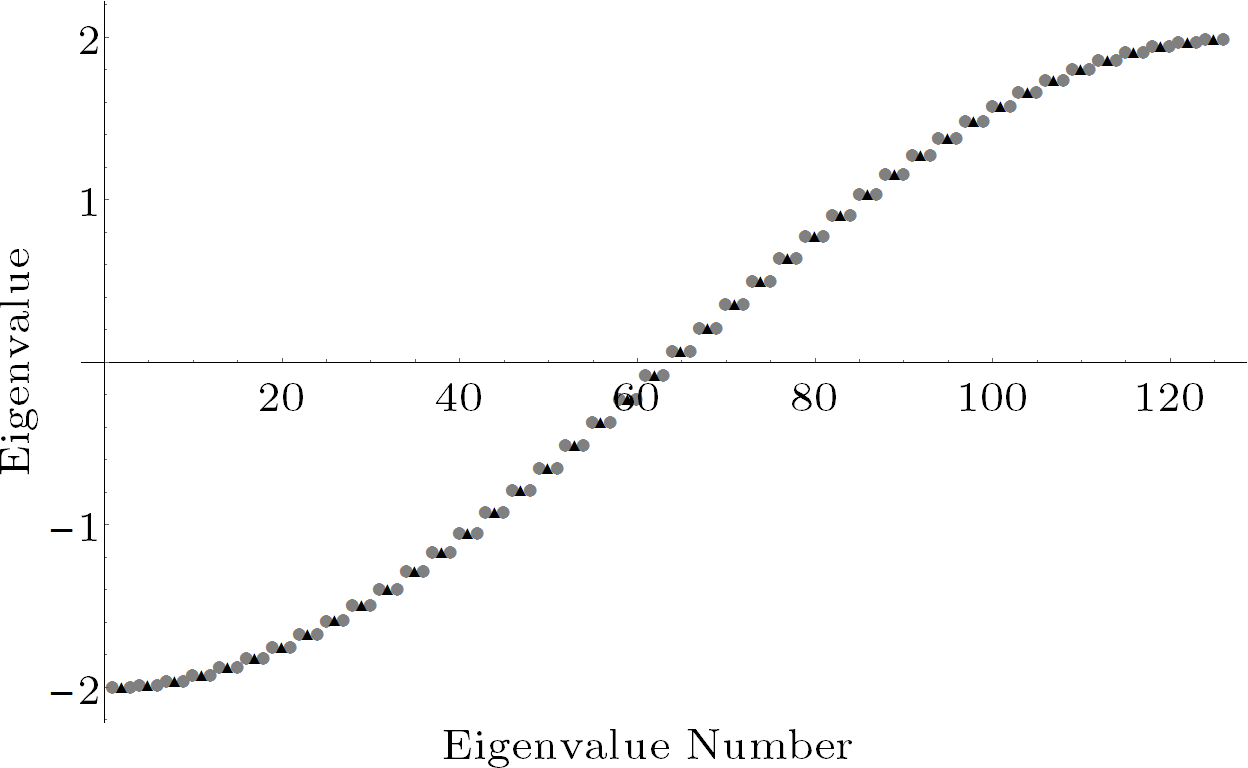}
\caption{New eigenvalue selection procedure implemented on a uniformly coupled chain of length 42. Controlled eigenvalues (grey circles) and uncontrolled (black triangles) are plotted.}\label{fig:uniform}
\end{figure}

Numerical results for a uniformly coupled chain were demonstrated in \cite{kay2022}. It behooves us to briefly revisit this case, as our strategy for selecting eigenvalues has now changed. While we previously achieved an almost perfectly linear spectrum, this was luck. If we apply our pair pinning strategy (Fig.\ \ref{fig:uniform}) for selecting eigenvalues, we can guarantee a solution, at the cost of a much higher state transfer time, $O(N^2)$ rather than $O(N)$. We see that controlled eigenvalues are chosen in pairs that are close to the eigenvalues of the original system. It appears that the uncontrolled eigenvalues are well approximated by the eigenvalues of the original system. For a chain, this is no surprise: having selected two eigenvalues of the same symmetry to pin the original eigenvalue, the strict interlacing of symmetric and antisymmetric eigenvalues means that there must be an antisymmetric eigenvalue pinned between those two points.

A natural generalisation to this is a chain where the couplings are not all equal. This is likely to be of greater practical interest. For example, we might randomly select couplings in the range $(J-\delta J,J+\delta J)$. We generally found that as $\delta J$ is increased if one tries to manually guess target eigenvalues, it becomes far more likely that, while $P$ and $Q$ can still be found, the eigenvalues do not interlace. Instead, we rely on the pair-pinning technique, with a consequential impact on transfer time.

To be explicit, we considered examples taken from the even sized, field-free case and:
\begin{itemize}
\item Calculated the eigenvalues $\Lambda$ of $C_+$, and the eigenvalues $\Gamma$ of its principal submatrix.
\item Evaluated $\delta_1=\min_{\lambda\in\Lambda,\gamma\in\Gamma}|\lambda-\gamma|$ and $\delta_2=\frac{1}{2}\min_{\lambda\in\Lambda,\gamma\in\Lambda}|\lambda+\gamma|$. From this, we calculated $\delta=\min(\delta_1,\delta_2)/2$.
\item For each $\lambda\in\Lambda$, selected a target eigenvalue in the range $(\lambda-2\delta,\lambda)$ and another in the range $(\lambda,\lambda+2\delta)$ satisfying the form $(2n+\frac{1}{2})\delta$.
\item Attempted to solve for a suitable extension. If the interlacing of roots $Q_A$ and $P_A$ was not satisfied, reduce $\delta$ and repeat.
  \end{itemize}
For all the random samples we tested (initial chain of 20 qubits, couplings chosen uniformly at random between $(0.9,1.1)$), reducing $\delta$ to $\delta=\min(\delta_1,\delta_2)/5$ successfully found a result (chain length 120, transfer time $\pi/\delta$). One such example is depicted in Fig.\ \ref{fig:nonuniform}.

\begin{figure}
\centering
\includegraphics[width=0.45\textwidth]{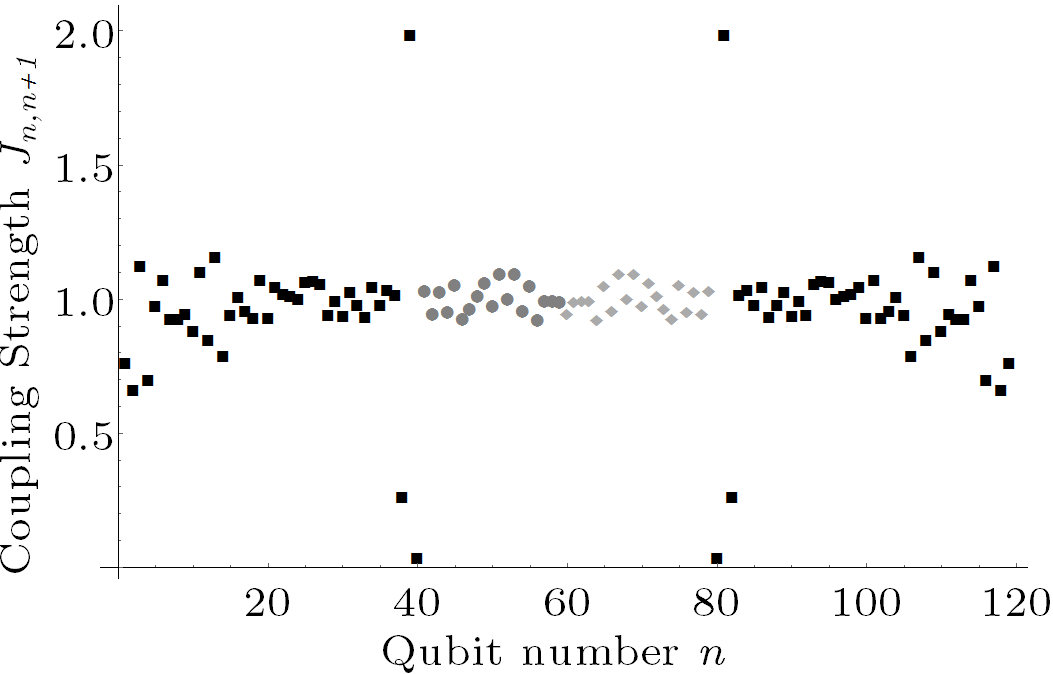}
\caption{A chain of 20 qubits with couplings randomly selected in the range $(0.9,1.1)$ (grey circles) can be symmetrised (diamonds) and extended (black squares), giving an overall chain of 120 qubits demonstrating perfect encoded transfer.}\label{fig:nonuniform}
\end{figure}

\subsubsection{Integer Selection}

There are surely better strategies that can help to optimise the transfer fidelity at shorter times. We did not explore these extensively. Nevertheless, note that there will be some original eigenvalues which are well separated compared to the scale $\delta$. As such, there may be many possible values of $n$ that could be chosen (even if our basic strategy requires us to select the one closest to the edge of the range). Given such a choice, it can help to select a value of $n$ with the greatest number of factors of 5 possible. This is because all cases which contain a factor of $5^k$ also satisfy the perfect transfer conditions at times $\pi/(5^k\delta)$. To see why this works, consider two different values of $\delta$, $\delta_X$ and $\delta_Y$. We note that if $\delta_Y=(4p+1)\delta_X$ for integer $p$, then
\begin{align*}
\left(2n\pm\frac{1}{2}\right)\delta_Y&=\left(2n\pm\frac{1}{2}\right)(4p+1)\delta_X \\
&=\left(2(4p+1)n\pm 2p\pm\frac12\right)\delta_X.
\end{align*}
Thus, $n'=(4p+1)n\pm p$. If some eigenvalues are of the form $(2n+\half)\delta_X$ and others are of the form $(2n+\half)\delta_Y=(2n'+\half)\delta_X$, then this shows that all eigenvectors satisfy the transfer conditions at $\pi/\delta_X$, while some also satisfy them at $\pi/\delta_Y$. Picking $p=1$ is the smallest case for which this works, giving the greatest range of options.


\subsection{Non-chain central region.}

All the examples that we have described so far just use a chain in the central region. We do not need to be restricted to this case, however. Any symmetric system will do. One advantage is that this may reduce the number of eigenvalues we have to consider due to the reduction to being fully supported (and therefore reduce the size of the required encoding region).

\begin{lemma}\label{lem:nosupport}
Let $\ket{\lambda}$ be an eigenvector of $C_+$ with eigenvalue $\lambda$ such that $\braket{1}{\lambda}=0$. Then $\ket{\lambda}$ is also an eigenvector of $H_+$ with eigenvalue $\lambda$.
\end{lemma}

The key point here is that such eigenvectors will have no support on the encoding and decoding regions. Hence, we do not have to prevent their use via encoding. We need only consider those eigenvectors of $C_+$ which have support on the first site. (To contrast with the case of the chain, all eigenvectors of a coupled chains have support on the first site.) As such, it is sufficient to reduce to fully supported systems.

Consider the example network depicted in Fig.\ \ref{fig:notchain}(a). We want to add an extension to vertices 1 and 6 to achieve (encoded) perfect transfer. We first note that there is a basis spanned by $(\ket{2}-\ket{3})/\sqrt{2}$ and $(\ket{3}-\ket{5})/\sqrt{2}$. Hence, we only need to consider the effective graph \cite{kay2005} depicted in Fig.\ \ref{fig:notchain}(b) in which we have excluded the two unsupported eigenvectors.
This is symmetric bipartite, and even. We choose to extend the system by 4 qubits at either end, and solve for the case that the symmetric subspace has eigenvalues $\frac{11}{4},\frac{7}{4},\frac{3}{4},-\frac{5}{4}$, which would yield perfect transfer in a time $2\pi$ by encoding across the first 5 spins. We find the solution displayed in Fig.\ \ref{fig:notchain}(c), which is readily verified to have the claimed eigenvalues, and a perfect transfer property.

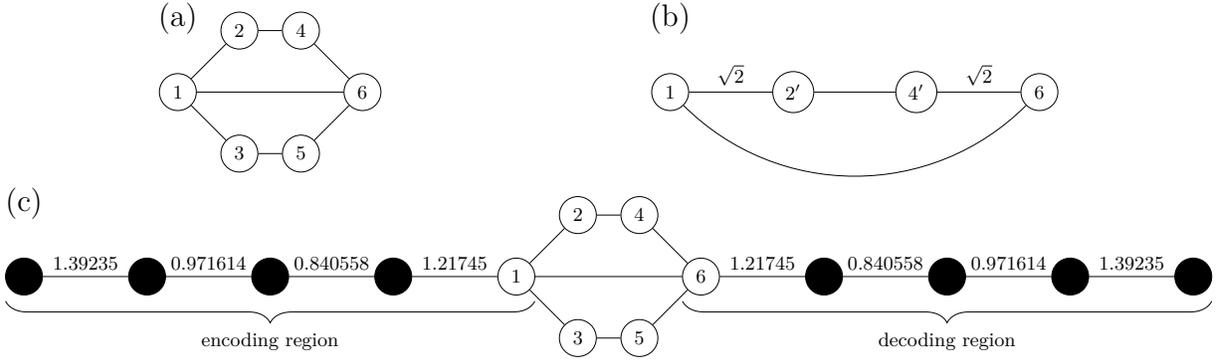
\begin{figure*}
\centering
\begin{adjustbox}{width=0.9\textwidth}
\begin{tikzpicture}[main/.style = {draw, circle,minimum width=0.6cm}] 
\begin{scope}
\node at (0,1.2) {\Large(a)};
\node[main] (1) at (0,0) {$1$}; 
\node[main] (2) at (1,1) {$2$}; 
\node[main] (3) at (1,-1) {$3$}; 
\node[main] (4) at (2,1) {$4$}; 
\node[main] (5) at (2,-1) {$5$}; 
\node[main] (6) at (3,0) {$6$}; 
\draw (1) -- (2) -- (4) -- (6) -- (5) -- (3) -- (1) -- (6);
\end{scope}
\begin{scope}[xshift=8cm]
\node at (0,1.2) {\Large(b)};
\node[main] (1) at (0,0) {$1$}; 
\node[main] (2) at (2,0) {$2'$}; 
\node[main] (3) at (4,0) {$4'$}; 
\node[main] (4) at (6,0) {$6$};  
\draw (1) -- node[midway, above] {$\sqrt{2}$}  (2) -- (3) -- node[midway, above] {$\sqrt{2}$} (4);
\draw (1) to [out=-45,in=-135,looseness=1] (4);
\end{scope}
\begin{scope}[yshift=-3cm,xshift=5.5cm]
\node at (-8,1.2) {\Large(c)};
\node[main,fill=black] (-4) at (-8,0) { };
\node[main,fill=black] (-3) at (-6,0) { };
\node[main,fill=black] (-2) at (-4,0) { };
\node[main,fill=black] (-1) at (-2,0) { };
\node[main] (1) at (0,0) {$1$}; 
\node[main] (2) at (1,1) {$2$}; 
\node[main] (3) at (1,-1) {$3$}; 
\node[main] (4) at (2,1) {$4$}; 
\node[main] (5) at (2,-1) {$5$}; 
\node[main] (6) at (3,0) {$6$}; 
\node[main,fill=black] (7) at (5,0) { };
\node[main,fill=black] (8) at (7,0) { };
\node[main,fill=black] (9) at (9,0) { };
\node[main,fill=black] (10) at (11,0) { };
\draw (1) -- (2) -- (4) -- (6) -- (5) -- (3) -- (1) -- (6);
\draw (-4) -- node[midway, above] {$1.39235$} (-3) -- node[midway, above] {$0.971614$} (-2) -- node[midway, above] {$0.840558$} (-1) -- node[midway, above] {$1.21745$} (1);
\draw (6) -- node[midway, above] {$1.21745$} (7) -- node[midway, above] {$0.840558$} (8) -- node[midway, above] {$0.971614$} (9) -- node[midway, above] {$1.39235$} (10);
\draw [decorate,decoration={brace,amplitude=10pt,mirror}]
(-8.3,-0.4) -- (0.3,-0.4) node [black,midway,below,yshift=-0.4cm] {encoding region};
\draw [decorate,decoration={brace,amplitude=10pt,mirror}]
(2.7,-0.4) -- (11.3,-0.4) node [black,midway,below,yshift=-0.4cm] {decoding region};
\end{scope}
\end{tikzpicture}
\end{adjustbox}
\caption{(a) Example network. Transfer between nodes 1 and 6 targetted. (b) After removal of two eigenvectors with no support on site 1, a fully supported network of 4 spins remains. (c) Extension with perfect encoded transfer.}\label{fig:notchain}
\end{figure*}

Let us now attempt to understand the importance of the requirement in Theorem \ref{thm:centralthesis} that $C_{\pm}$ should not share any eigenvalues. When this requirement is not fulfilled, it may still be possible to extend the system. For example, consider the case depicted in Fig.\ \ref{fig:transferthroughnull}\footnote{This example has been chosen to be small. It does not feature the symmetrisation procedure of $B\mapsto C$. If it did, there would be a parameter $J'$, and the only challenge is if $C_{\pm}$ share an eigenvalues for all values $J'$.}. We can subdivide this using the symmetry to give
$$
C_+=\left(\begin{array}{ccc}
0 & \sqrt{2} & 0 \\
\sqrt{2} & 0 & \sqrt{2} \\
0 & \sqrt{2} & 0 \end{array}\right),\qquad C_-=\left(\begin{array}{c} 0 \end{array}\right).
$$
Both components have a 0 eigenvalue. (There is also a component that we have neglected via Lemma \ref{lem:nosupport}.) We can still proceed to choose eigenvalues that guarantee real roots for $Q_A(z)$ and $P_A(z)$, and hope to get lucky over the interlacing. In this instance, we assign eigenvalues $\Lambda^+=\{\frac{60}{41},\frac{80}{41},\frac{100}{41}\}$ and $\Lambda^-=\{\frac{70}{41}\}$, and the solution is also given in Fig.\ \ref{fig:transferthroughnull}. There is perfect encoded transfer in time $41\pi/10$. The limiting factor in terms of transfer time is having to squeeze three eigenvalues in the range $z\in(\sqrt{2},2)$ (the only region $\mu_{B^+}(z)\mu_{B^-}(z)<0$ that allows us to pack as many eigenvalues as we need).

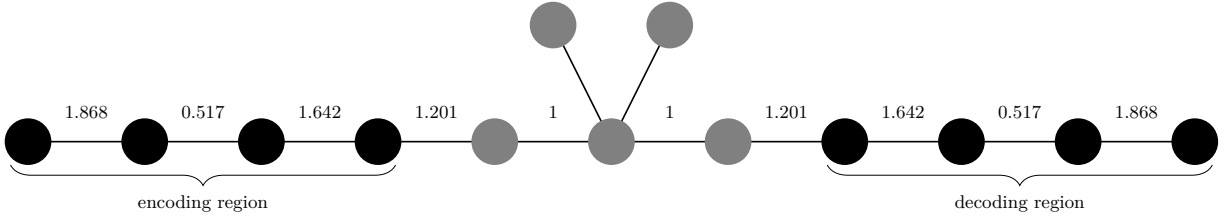
\begin{figure*}
\centering
\begin{adjustbox}{width=0.9\textwidth}
\begin{tikzpicture}
\node at (-9,0.5) {1.868};
\node at (-7,0.5) {0.517};
\node at (-5,0.5) {1.642};
\node at (-3,0.5) {1.201};
\node at (-1,0.5) {1};
\node at (9,0.5) {1.868};
\node at (7,0.5) {0.517};
\node at (5,0.5) {1.642};
\node at (3,0.5) {1.201};
\node at (1,0.5) {1};
\draw [thick] (-1,2) -- (0,0) -- (1,2);
\draw [thick] (-10,0) -- (10,0);
\node[circle,style={fill=gray,minimum width=0.8cm,text=white}] at (0,0) {};
\node[circle,style={fill=gray,minimum width=0.8cm,text=white}] at (2,0) {};
\node[circle,style={fill=gray,minimum width=0.8cm,text=white}] at (-2,0) {};
\node[circle,style={fill=black,minimum width=0.8cm,text=white}] at (4,0) {};
\node[circle,style={fill=black,minimum width=0.8cm,text=white}] at (-4,0) {};
\node[circle,style={fill=black,minimum width=0.8cm,text=white}] at (6,0) {};
\node[circle,style={fill=black,minimum width=0.8cm,text=white}] at (-6,0) {};
\node[circle,style={fill=black,minimum width=0.8cm,text=white}] at (8,0) {};
\node[circle,style={fill=black,minimum width=0.8cm,text=white}] at (-8,0) {};
\node[circle,style={fill=black,minimum width=0.8cm,text=white}] at (10,0) {};
\node[circle,style={fill=black,minimum width=0.8cm,text=white}] at (-10,0) {};
\node[circle,style={fill=gray,minimum width=0.8cm,text=white}] at (1,2) {};
\node[circle,style={fill=gray,minimum width=0.8cm,text=white}] at (-1,2) {};

\draw [decorate,decoration={brace,amplitude=10pt,mirror}]
(-10.3,-0.4) -- (-3.7,-0.4) node [black,midway,below,yshift=-0.4cm] {encoding region};
\draw [decorate,decoration={brace,amplitude=10pt,mirror}]
(3.7,-0.4) -- (10.3,-0.4) node [black,midway,below,yshift=-0.4cm] {decoding region};
\end{tikzpicture}
\end{adjustbox}
\caption{Uniformly coupled network (grey qubits) extended by black qubits to create perfect encoded transfer.}\label{fig:transferthroughnull}
\end{figure*}

In the previous example, the coincidence of eigenvalues between $C_{\pm}$ was not prohibitive. Nevertheless, it can be an important indication of the impossibility of transfer. Consider the following example, inspired by \cite{pemberton-ross2011},
\begin{center}
\begin{tikzpicture}
\draw (0,0) -- node[above,pos=0.5] {$1$} (1,1) --  node[above,pos=0.5] {$-1$} (2,0) -- node[above,pos=0.5] {$1$} (3,-1) -- node[above,pos=0.5] {$1$} (4,0) -- node[above,pos=0.5] {$1$} (3,1) -- node[above,pos=0.5] {$-1$} (2,0) -- node[above,pos=0.5] {$1$} (1,-1) -- node[above,pos=0.5] {$1$} (0,0);
\node[circle,style={fill=gray,minimum width=0.3cm,text=white}] at (0,0) {};
\node[circle,style={fill=black,minimum width=0.3cm,text=white}] at (2,0) {};
\node[circle,style={fill=gray,minimum width=0.3cm,text=white}] at (4,0) {};
\node[circle,style={fill=black,minimum width=0.3cm,text=white}] at (1,1) {};
\node[circle,style={fill=black,minimum width=0.3cm,text=white}] at (3,1) {};
\node[circle,style={fill=black,minimum width=0.3cm,text=white}] at (1,-1) {};
\node[circle,style={fill=black,minimum width=0.3cm,text=white}] at (3,-1) {};
\end{tikzpicture}
\end{center}
where we would hope to transfer between the two extremal (grey) qubits. In fact, this is impossible --- if an excitation is on the first site, it becomes a uniform superposition over its two neighbours. However, this is a `dark state' \cite{pemberton-ross2010,pemberton-ross2011} that cannot transfer any further to the right. We see this in the matrices $C^{\pm}$, which each decouple into two subspaces. They have one of those subspaces in common, and hence have common eigenvalues.
\begin{center}
\begin{tikzpicture}
\node at (-1,0) {$C^{\pm}=$};
\draw (0,0) -- node[above,pos=0.5] {1} (1,1) --  node[above,pos=0.5] {-1} (2,0) -- node[above,pos=0.5] {1} (1,-1) -- node[above,pos=0.5] {1} (0,0);
\node at (3,0) {$\pm 1\qquad\equiv$};
\node[circle,style={fill=black,minimum width=0.3cm,text=white}] at (0,0) {};
\node[circle,style={fill=black,minimum width=0.3cm,text=white}] at (2,0) {};
\node[circle,style={fill=black,minimum width=0.3cm,text=white}] at (1,1) {};
\node[circle,style={fill=black,minimum width=0.3cm,text=white}] at (1,-1) {};
\draw (4.5,0.5) --node[above,pos=0.5] {$\sqrt{2}$} (6,0.5);
\draw (4.5,-0.5) --node[above,pos=0.5] {$\sqrt{2}$} (6,-0.5);
\node[circle,style={fill=black,minimum width=0.3cm,text=white}] at (4.5,0.5) {};
\node[circle,style={fill=black,minimum width=0.3cm,text=white}] at (4.5,-0.5) {};
\node[circle,style={fill=black,minimum width=0.3cm,text=white}] at (6,0.5) {};
\node[circle,style={fill=black,minimum width=0.3cm,text=white}] at (6,-0.5) {};
\node at (6,0.8) {$\pm 1$};
\end{tikzpicture}
\end{center}

\section{Creating States}

Perfect transfer is just one task that we might want to achieve. Another option, as suggested in \cite{kay2022}, is the creation of an arbitrary single-excitation state within the fixed region that we were initially given. In principle, this is quite straightforward --- if we consider $C$ to be two distinct halves of the original $B$ and its mirror, and if the entire system were a perfect transfer system (rather than just a subset of the eigenvalues satisfying those conditions), then we could create any state we wanted on the mirror component, and after the perfect transfer time, it will appear on the main component, perfectly.

Depending on how many eigenvalues we set (and hence how far we extend the system at either end), we control a varying proportion of the eigenvalues, and hence control how close to that perfect condition we are. For instance, if we use the same extension as we would for perfect encoded state transfer, we control $2/3$ of the eigenvalues, and high fidelity transfer results. Of course, in systems where we have made the reduction to fully supported, none of those eliminated eigenvectors in the full system can ever be populated. This is not a concern for chains, where no such reduction is necessary.

Specifically, let $H$ be the full solution for perfect encoded transfer, and $P_\text{out}$ and $P_\text{in}$ are the projectors on the the original ($B$) system, and the component to the right of the $B$ system (i.e.\ the symmetrised $B$ and the extension on that side). If we want to create a state $\ket{\Psi}$, then the best input state to use is
$$
P_\text{in}e^{iHt_0}P_\text{out}\ket{\Psi}.
$$
As such, the worst performance for any target state can be assessed via the minimum singular value of $P_\text{in}e^{iHt_0}P_\text{out}$. See Fig.\ \ref{fig:schematic2}.

For uniformly coupled chains, we assessed this in \cite{kay2022}. State creation fidelity had the potential to be high, but there was generally a small space of low fidelity. It turns out that the construction of Fig.\ \ref{fig:uniform} via pair-pinning, while much slower, gives vastly enhanced fidelities. The specific case of Fig.\ \ref{fig:uniform} has a smallest singular value of $1-5\times 10^{-6}$. The results for random chains are comparable.

We view this as an artefact of the strategy for choosing eigenvalues: we would be able to perfectly create any state that is entirely supported on the eigenvectors of $\Gamma_P$. The error is thus due to any component of an eigenvector of $\Gamma_{\bar P}$ which cannot be replicated by a linear combination of eigenvectors is $\Gamma_P$. Our pair-pinning strategy takes each eigenvalue $\eta$ of $C_{\sigma}$ and assigns a flanking pair $\eta\pm\epsilon_{\pm}$. From Lemma \ref{lem:evector inverse}, the eigenvector restricted to $\mathcal{B}$, $\ket{b}$, is
$$
\ket{b}\propto (C_{\pm}-\lambda\identity)^{-1}\ket{1}.
$$
One can readily show that
$$
\ket{\eta}=\frac{\epsilon_-}{\epsilon_--\epsilon_+}\ket{b(\epsilon_+)}-\frac{\epsilon_+}{\epsilon_--\epsilon_+}\ket{b(-\epsilon_-)}+O(\epsilon^2).
$$
Since any target state can be written in terms of the spectral basis of $C_{\sigma}$, there is at worst an $O(\epsilon^2)$ error expressing it in terms of the controlled eigenvectors. Hence, at the cost of increasing time, we can make arbitrarily accurate copies of any target state.

While arbitrarily high fidelity state creation of a state $\ket{\eta}$ is always possible, perfect state creation is generically impossible. In the typical case, none of the $2N_B$ uncontrolled eigenvalues exactly satisfy the perfect transfer conditions. Assume that half of these are symmetric, and half antisymmetric. For any $\lambda\in\Gamma^{\sigma}_{\bar P}$, if perfect state creation is possible, it must be that $\bra{\eta}P_\text{out}\ket{\lambda}=0$. Hence, the only states that can be made perfectly are those in the null space of the square matrix
$$
M_{\sigma}=\sum_{k=1}^{|\Gamma^{\sigma}_{\bar P}|}\ket{k}\bra{\lambda_k}P_{\text{out}}
$$
Since we are working with fully supported systems, $\lambda_k$ is not an eigenvalue of $C_{\sigma}$. Thus, by Lemma \ref{lem:evector inverse}, $P_\text{out}\ket{\lambda_k}\propto(C_{\sigma}-\lambda_k\identity)^{-1}\ket{1}$. Our claim is that $M_{\sigma}$ does not have a null space, i.e.\ its determinant is non-trivial.
Let $U$ be the unitary that diagonalises $C_\sigma$:
$$
C_\sigma=UDU^\dagger, \qquad D=\sum_i\eta_i\proj{i}.
$$
Then
$$
|\text{det}(M)|=\left|\text{det}\sum_{i,j}\ket{i}\bra{j}\frac{\bra{i}U^\dagger\ket{1}}{\eta_i-\lambda_j}\right|.
$$
By the assumption of full support, $\bra{i}U^\dagger\ket{1}\neq0$. Each row has a constant, non-zero factor $\bra{i}U^\dagger\ket{1}$ which we can hence pull out. Thus, the determinant is only zero if
$$
\sum_{i,j}\frac{1}{\eta_i-\lambda_j}\ket{i}\bra{j}
$$
has determinant 0. However, this is just
$$
\left|\frac{\prod_{i\neq j}(\lambda_i-\lambda_j)\prod_{i\neq j}(\eta_i-\eta_j)}{\prod_{i,j}(\lambda_i-\eta_j)}\right|.
$$
Since all the eigenvalues are unique, it is never 0. $M$ is full rank, and no state that can be made perfectly.

\subsection{Time Control}

\begin{figure}
\centering
\includegraphics[width=0.4\textwidth]{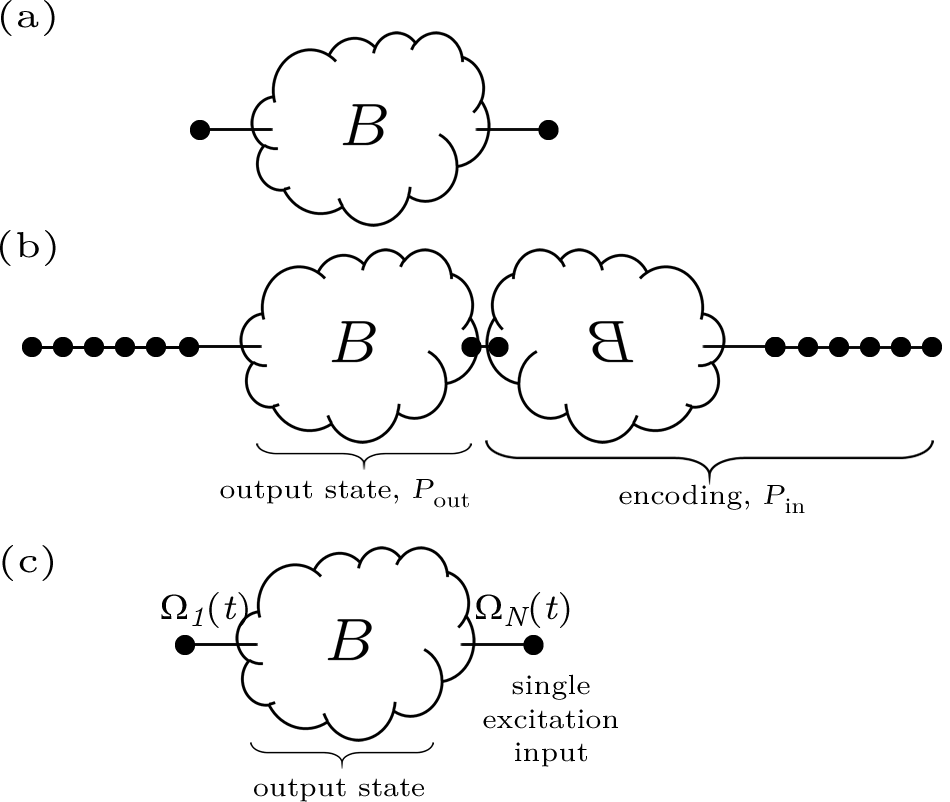}
\caption{(a) Given a fixed Hamiltonian, $B$, our aim is to create a state in the bulk. (b) We first symmetrise the system $B$ about the output vertex, extending the input vertex and its mirror with chains, and encoding/decoding across the added regions. (c) The state creation protocol can be replicated with time control of the extremal coupling strengths \cite{haselgrove2005}.}\label{fig:schematic2}
\end{figure}

High quality state creation might sound rather pointless: being able to create the right target state on the extension is surely no easier than creating the target state directly on the original spins? However, consider making all the extended spins `virtual' via the time-control technique of Haselgrove \cite{haselgrove2005}, as shown in Fig.\ \ref{fig:schematic2}. Now, whatever state you want to create, you just start with a single excitation localised on a single spin, and your time control determines what state is created.

\subsection{Multiple Excitation State Creation}

High quality state creation in the single excitation subspace also bodes well for multiple-excitation state creation (but note that this only applies to chains, and that the time-control technique does not apply directly). For example, if we wanted to create a GHZ state $\ket{GHZ}=(\ket{0}^{\otimes N_B}+\ket{1}^{\otimes N_B})/\sqrt{2}$ across all the sites $\mathcal{B}$, we can achieve a fidelity of
$$
F=\bra{GHZ}\rho\ket{GHZ}=\frac14\left(1+\prod_i\lambda_i\right)^2
$$
where $\lambda_i$ are the singular values of $P_\text{in}e^{iHt_0}P_\text{out}$ \cite{haselgrove2005}. The uniformly coupled example returns a fidelity of 0.9998. Random cases have similar levels of performance. The specific case of Fig.\ \ref{fig:nonuniform} has $F=1-6\times 10^{-6}$.

\section{Conclusions \& Future Work}

We have shown how a fixed quantum system of effective size $N$ can be symmetrically extended, first by adding a mirror of the original system, and then by adding chains at either end. With a total system size of $6N$, it is possible to achieve perfect transfer with encoding over the extensions. As demonstrated in \cite{kay2022}, one can reduce the sizes of extension, reducing the transfer time, and still achieving high quality transfer. Moreover, it is easier to fulfil the conditions on eigenvalue interlacing. While \cite{kay2022} focussed entirely on the case of extending a uniformly coupled chain, and gave no justification for the existence of solutions, the results of this paper show that almost any quantum network can support these results. The only requirement is that the two symmetry subspaces do not share any eigenvalue, as this detects the possibility that transfer could be impossible. Moreover, our pair-pinning strategy gives a constructive method that guarantees, in a certain limit, that solutions exist. It suggests that the state transfer time in this scenario is largely governed by the inverse of the spectral gap of the original system. The extensions can be replaced by individual, time-controlled couplings \cite{haselgrove2005}, bringing us very close to a realistic experimental scenario.

We have also proven that any single-excitation state can be created on the output region, with arbitrary accuracy. On chains, this includes arbitrary states comprising multiple excitations.

As a next step, it will be interesting to apply these results to systems such as the IBM Q devices using OpenPulse. However, the results are not immediately applicable because, while those systems are primarily defined by a fixed Hamiltonian with a coupling Hamiltonian which is exactly of the form we desire, the time control that we can access is only at the individual qubit level. It should be possible to simulate our scheme with the control available, but it requires further investigation.


%

\end{document}